\documentclass{article}
\pdfoutput=1

\usepackage[utf8]{inputenc}
\usepackage[T1]{fontenc}
\usepackage[english]{babel}
\usepackage{fullpage}

\usepackage{amsmath, amsthm, amssymb}
\usepackage{enumerate}
\usepackage{authblk}
\usepackage{thmtools}
\usepackage{mathrsfs}
\usepackage{thm-restate}
\usepackage{soul}
\usepackage[dvipsnames]{xcolor} 
\usepackage{tikz}
\usepackage[labelsep=period,labelfont=bf]{caption}
\usepackage{float,graphicx,subcaption}
\usepackage{fullpage}
\usepackage{thm-restate}
\usepackage{hyperref}
\usepackage{comment}
\usepackage{cleveref}

\theoremstyle{definition}
\newtheorem{theorem}{Theorem}[section]

\newtheorem{corollary}[theorem]{Corollary}

\newtheorem{lemma}[theorem]{Lemma}

\theoremstyle{remark}

\def\C{\mathcal{C}} 
\def\cD{\mathcal{D}} 

\def\Oh{O} 

\usetikzlibrary{matrix,decorations.pathreplacing, calc, positioning,fit,shapes.geometric}
\usetikzlibrary{intersections}
\usetikzlibrary{patterns}
\usetikzlibrary{patterns.meta}
\usetikzlibrary{arrows.meta}
\usetikzlibrary{shapes.arrows}
\usetikzlibrary{calc}
\usetikzlibrary{shadows.blur}

\def\centerarc[#1](#2)(#3:#4:#5)
    { \draw[#1] ($(#2)+({#5*cos(#3)},{#5*sin(#3)})$) arc (#3:#4:#5) }
\def\nodearc(#1)(#2)(#3:#4)[#5](#6)
    { \node (#1) at ($(#2)+({#4*cos(#3)},{#4*sin(#3)})$) [#5] {#6}}
\def\nodeellipse(#1)(#2)(#3:#4:#5)[#6](#7)
    { \node (#1) at ($(#2)+({#4*cos(#3)},{#5*sin(#3)})$) [#6] {#7}}

\pgfdeclarepattern{
  name=hatch,
  parameters={\hatchsize,\hatchangle,\hatchlinewidth},
  bottom left={\pgfpoint{-.1pt}{-.1pt}},
  top right={\pgfpoint{\hatchsize+.1pt}{\hatchsize+.1pt}},
  tile size={\pgfpoint{\hatchsize}{\hatchsize}},
  tile transformation={\pgftransformrotate{\hatchangle}},
  code={
    \pgfsetlinewidth{\hatchlinewidth}
    \pgfpathmoveto{\pgfpoint{-.1pt}{-.1pt}}
    \pgfpathlineto{\pgfpoint{\hatchsize+.1pt}{\hatchsize+.1pt}}
    \pgfusepath{stroke}
  }
}

\tikzset{
  hatch size/.store in=\hatchsize,
  hatch angle/.store in=\hatchangle,
  hatch line width/.store in=\hatchlinewidth,
  hatch size=5pt,
  hatch angle=0pt,
  hatch line width=.5pt,
}

\def\dash{10pt}

\tikzset{
    dashone/.style={dash pattern=on 5pt off 15pt},
  }

\newcommand{\colora}{Gray}
\newcommand{\colorb}{SkyBlue}
\newcommand{\colorc}{RawSienna}
\newcommand{\colord}{YellowOrange}
\newcommand{\colore}{WildStrawberry}
\newcommand{\colorf}{ForestGreen}
\newcommand{\colorg}{Violet}

\title{An algorithmic Vizing’s theorem: toward efficient edge-coloring sampling with an optimal number of colors}

\author[1]{Lucas De Meyer}
\author[2,3]{František Kardoš}
\author[4]{Aurélie Lagoutte}
\author[5,6]{Guillem Perarnau}
\affil[1]{Universite Claude Bernard Lyon 1, CNRS, INSA Lyon, LIRIS, UMR 5205, 69622 Villeurbanne, France}
\affil[2]{Univ. Bordeaux, CNRS, Bordeaux INP, LaBRI, UMR 5800, 33400 Talence, France}
\affil[3]{Departement of Computer Science, Comenius University, Bratislava, Slovakia}
\affil[4]{Univ. Grenoble Alpes, CNRS, Grenoble INP, G-SCOP, 38000 Grenoble, France}
\affil[5]{Universitat Polit\`ecnica de Catalunya (UPC), Barcelona, Spain.}
\affil[6]{Centre de Recerca Matem\`atica, Bellaterra, Spain.}

\date{January 20, 2025}

\begin{document}

\maketitle

\footnotetext[0]{This work arise from a workshop supported by ANR project GrR (ANR-18-CE40-0032).
FK has been supported by APVV grant APVV-23-0076 and by VEGA grant VEGA 1/0727/22. 
GP has been supported by the grants RED2022-134947-T, PID2023-147202NB-I00 and CEX2020-001084-M, all of them funded by MICIU/AEI/10.13039/501100011033.}

\begin{abstract}
The problem of sampling edge-colorings of graphs with maximum degree $\Delta$ has received considerable attention and efficient algorithms are available when the number of colors is large enough with respect to $\Delta$. Vizing's theorem guarantees the existence of a $(\Delta+1)$-edge-coloring, raising the natural question of how to efficiently sample such edge-colorings. In this paper, we take an initial step toward addressing this question. Building on the approach of Dotan, Linial, and Peled~\cite{Dotan20}, we analyze a randomized algorithm for generating random proper $(\Delta+1)$-edge-colorings, which in particular provides an algorithmic interpretation of Vizing's theorem. The idea is to start from an arbitrary non-proper edge-coloring with the desired number of colors and at each step, recolor one edge uniformly at random provided it does not increase the number of conflicting edges (a potential function will count the number of pairs of adjacent edges of the same color). We show that the algorithm almost surely produces a proper $(\Delta+1)$-edge-coloring and propose several conjectures regarding its efficiency and the uniformity of the sampled colorings.
\end{abstract}

\section{Introduction}

The problem of sampling discrete combinatorial structures is one of the central topics in modern Theoretical Computer Science, with close connections to Discrete Mathematics, Probability Theory and Statistical Physics. A first motivation stems from the analysis of complex probability distributions, intractable from the analytic point of view, but that can be generated via Markov Chain Monte Carlo (MCMC) algorithms~\cite{Mar,saloff1997lectures} whose transitions are given by local transformations. This is intimately related to the topic of mixing times in Markov Chains. Another motivation comes from the correspondence between approximate counting and approximate sampling; thus, an efficient almost-uniform sampler gives rise to an efficient approximation scheme, and this is how many of the strongest approximate counting results are obtained. Exact counting of such discrete structures is often a computationally hard algorithmic problem (\#P-hard).

A paradigmatic instance is sampling proper vertex-colorings of a graph, for which there is an extensive and rich literature. Given a graph $G$ on $n$ vertices and maximum degree $\Delta=\Delta(G)$, the aim is to design an efficient algorithm that samples proper $k$-colorings of $G$ according to a distribution that is uniform, or close to it (according to a certain distance). Here, efficient means that it is polynomial in $n$ and $\log{(1/\varepsilon)}$, where $\varepsilon>0$ is a bound on the distance between the distribution of the sample and the uniform one. 
MCMC algorithms, such as Glauber dynamics, have demonstrated an excellent performance for this problem~\cite{Chen19,Jerrum95,Vigoda00}, and to date, they outperform the best deterministic approaches to approximate the counting of proper colorings~\cite{Liu19}. It is conjectured that the dynamics can efficiently sample random colorings as long as $k\geq \Delta+2$, beyond this point the chain fails to be ergodic. Currently, this is known for $k\geq 1.809\Delta$~\cite{Carlson24}.

The case of edge-colorings, which is the subject matter of this article, is much less studied. Proper edge-colorings are fundamental structures in graph theory, and many combinatorial designs can be represented using edge colorings, for instance: (i) \emph{Latin squares} of size $n$ correspond to proper $n$-edge-colorings of the complete bipartite graph $K_{n,n}$, (ii) \emph{one-factorizations} of size $2n-1$ correspond to $(2n-1)$-edge-colorings of the complete graph $K_{2n}$, and (iii)
\emph{resolvable Steiner systems} with parameters $(r,r,n)$ correspond to proper edge-colorings of the complete $r$-uniform hypergraph $K_n^r$. The quest for efficient generation of random combinatorial designs comes in hand with the efficient sampling of proper edge-colorings.

Vizing's theorem~\cite{Vizing65} states that the minimum number of colors required to properly color the edges of a graph, is either $\Delta$ or $\Delta+1$ colors. Accordingly, this splits the class of all graphs into Class~I and Class~II graphs, and it is NP-hard to decide in which class a  graph is~\cite{holyer1981np}. In particular, it poses the natural question of whether there is an efficient procedure to sample proper $k$-edge-colorings for $k\geq \Delta+1$, which we call the \emph{algorithmic Vizing's problem}. (Note that, in Class I, the dynamics is not ergodic for $k=\Delta$ as the colors of the edges around a vertex of maximum degree are frozen.) 

Again, the MCMC approach is a natural candidate. Wang, Zhang and Zhang~\cite{Wang23} showed that the Glauber dynamics is an efficient sampler for proper $k$-edge-colorings provided that $k\geq (2+\varepsilon)\Delta$, for any $\varepsilon>0$. In other words, it is an efficient sampler for proper $k$-vertex-colorings of any \emph{line} graph $H$, if $k\geq (1+\varepsilon)\Delta(H)$. This improves upon the general bound~\cite{Carlson24} for graphs in this hereditary class. Remarkably, the result is asymptotically optimal: the Glauber dynamic on $k$-edge-colorings of Kneser graphs with $k= 2\Delta-1$ is not ergodic~\cite{Knesser}, see~\Cref{fig:kneser}. 

\begin{figure}[hbtp]
        \begin{center}
        \tikzstyle{vertex}=[circle,draw, minimum size=7pt, scale=0.7, inner sep=1pt, fill = black]
        \tikzstyle{fleche}=[->,>=latex]
        \tikzstyle{labell}=[opacity=0,text opacity=1, scale =1.2]
        \begin{tikzpicture}[scale=0.8]

        \node (a1) at (90:2cm) [vertex] {};
        \node (a2) at (90:1cm) [vertex] {};
        \node (a3) at (162:2cm) [vertex] {};
        \node (a4) at (162:1cm) [vertex] {};
        \node (a5) at (234:2cm) [vertex] {};
        \node (a6) at (234:1cm) [vertex] {};
        \node (a7) at (306:2cm) [vertex] {};
        \node (a8) at (306:1cm) [vertex] {};
        \node (a9) at (18:2cm) [vertex] {};
        \node (a10) at (18:1cm) [vertex] {};
        
        \draw[line width= 1.5, \colorb] (a1) to (a3);
        \draw[line width= 1.5, \colorf] (a1) to (a9);
        \draw[line width= 1.5, \colorc] (a3) to (a5);
        \draw[line width= 1.5, \colord] (a5) to (a7);
        \draw[line width= 1.5, \colorg] (a7) to (a9);
        
        \draw[line width= 1.5, \colord] (a1) to (a2);
        \draw[line width= 1.5, \colorg] (a3) to (a4);
        \draw[line width= 1.5, \colorf] (a6) to (a5);
        \draw[line width= 1.5, \colorb] (a8) to (a7);
        \draw[line width= 1.5, \colorc] (a10) to (a9);
        
        \draw[line width= 1.5, \colorg] (a6) to (a2);
        \draw[line width= 1.5, \colorc] (a2) to (a8);
        \draw[line width= 1.5, \colord] (a4) to (a10);
        \draw[line width= 1.5, \colorf] (a8) to (a4);
        \draw[line width= 1.5, \colorb] (a10) to (a6);

        \end{tikzpicture}
        \end{center}
        \caption{A frozen $(2\Delta - 1)$-edge-coloring of a Kneser graph}
        \label{fig:kneser}
\end{figure}
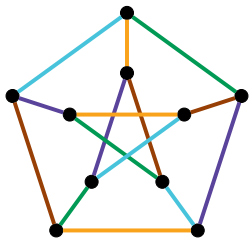

In order to break the $2\Delta$ barrier, there are two approaches: either we restrict the class of graphs of interest, or, we devise a new algorithm that bypasses the ergodicity issue in Glauber dynamics. Within the first line of research, Delcourt, Heinrich and Perarnau~\cite{Delcourt20} proved that the dynamics efficiently generates proper $k$-edge-coloring of trees for $k\geq \Delta+1$; see also \cite{Carlson2024b} for a more efficient algorithm. 

Within the second line of research, very recently, Carlson et al.~\cite{Carlson2024b} posed the question of whether there is an MCMC algorithm (or any other method) to sample proper edge-colorings using less than $2\Delta$ colors, or, alternatively, whether there is a hardness of approximate counting result in this regime. Here the main obstacle to overcome is the constrained nature of edge-colorings with few colors; indeed local transformations might not even be possible without violating the constraints. A good example is the MCMC algorithm of Jacobson and Matthews~\cite{JacMat} to sample Latin squares (equivalently proper $n$-edge-colorings of $K_{n,n}$), which outputs an almost uniform random Latin square by navigating through a much larger set of combinatorial objects; up to this date, it is not known if this algorithm is efficient. 

Dota, Linial and Peled~\cite{Dotan20} proposed an algorithm to generate uniform random one-factorizations of the complete graph with an even number of vertices. The \emph{randomized hill-climbing method} is based on a Markov Chain on edge-partitions (here seen as edge-colorings) whose trajectories are governed by local transformations that decrease a defined potential, whose minimum value is attained at one-factorizations (here, proper edge-colorings). 
The goal of this paper is to get further insights on this method and extend its applicability to general graphs.

For any (non-necessarily proper) $k$-edge-coloring $\sigma:E(G)\to [k]=\{1,\dots, k\}$, define the \emph{potential} $\psi$ of $\sigma$ as
\begin{equation}\label{eq:potential}
\psi(\sigma) = \sum_{v\in V(G)}\sum_{\alpha\in [k]} \binom{d_{\sigma,\alpha}(v)}{2},
\end{equation}
where $d_{\sigma,\alpha}(v)=|\{u\in V(G): uv\in E(G), \sigma(uv)=\alpha\}|$ is the \emph{$\alpha$-degree} of $v$ in $\sigma$. In other words, it is the number of pairs of adjacent edges that are equally colored (see \Cref{fig:pot}). Indeed, the minimum value of this potential is $0$ and it is attained if and only if the edge-coloring is proper.

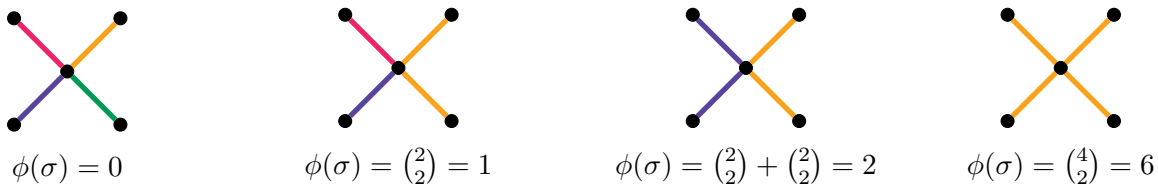
\begin{figure}[hbtp]
        \tikzstyle{vertex}=[circle,draw, minimum size=7pt, scale=0.7, inner sep=1pt, fill = black]
        \tikzstyle{fleche}=[->,>=latex]
        \tikzstyle{labell}=[opacity=0,text opacity=1, scale =1.1]

        \begin{subfigure}[t]{.2\textwidth}
            \centering
            \begin{tikzpicture}[scale=1]
                \node (a1) at (0, 0) [vertex] {};
                \node (a2) at (45:1cm) [vertex] {};
                \node (a3) at (-45:1cm) [vertex] {};
                \node (a4) at (-135:1cm) [vertex] {};
                \node (a5) at (135:1cm) [vertex] {};
             
                \draw[line width= 2, \colord] (a1) to (a2);
                \draw[line width= 2, \colorf] (a1) to (a3);
                \draw[line width= 2, \colorg] (a1) to (a4);        
                \draw[line width= 2, \colore] (a1) to (a5);

                \node (h4) at (0,-1.3) [labell] {$\phi(\sigma) = 0$};
                
            \end{tikzpicture}
        \end{subfigure}
        \hfill
        \begin{subfigure}[t]{.2\textwidth}
            \centering
            \begin{tikzpicture}[scale=1]
            
                \node (a1) at (0, 0) [vertex] {};
                \node (a2) at (45:1cm) [vertex] {};
                \node (a3) at (-45:1cm) [vertex] {};
                \node (a4) at (-135:1cm) [vertex] {};
                \node (a5) at (135:1cm) [vertex] {};
             
                \draw[line width= 2, \colord] (a1) to (a2);
                \draw[line width= 2, \colord] (a1) to (a3);
                \draw[line width= 2, \colorg] (a1) to (a4);        
                \draw[line width= 2, \colore] (a1) to (a5);

                \node (h4) at (0,-1.3) [labell] {$\phi(\sigma) = \binom{2}{2} = 1$};
                
            \end{tikzpicture}
        \end{subfigure}
        \hfill
        \begin{subfigure}[t]{.2\textwidth}
            \centering
            \begin{tikzpicture}[scale=1]
            
                \node (a1) at (0, 0) [vertex] {};
                \node (a2) at (45:1cm) [vertex] {};
                \node (a3) at (-45:1cm) [vertex] {};
                \node (a4) at (-135:1cm) [vertex] {};
                \node (a5) at (135:1cm) [vertex] {};
             
                \draw[line width= 2, \colord] (a1) to (a2);
                \draw[line width= 2, \colord] (a1) to (a3);
                \draw[line width= 2, \colorg] (a1) to (a4);        
                \draw[line width= 2, \colorg] (a1) to (a5);

                \node (h4) at (0,-1.3) [labell] {$\phi(\sigma) = \binom{2}{2} + \binom{2}{2} = 2$};
                
            \end{tikzpicture}
        \end{subfigure}
        \hfill
        \begin{subfigure}[t]{.2\textwidth}
            \centering
            \begin{tikzpicture}[scale=1]
            
                \node (a1) at (0, 0) [vertex] {};
                \node (a2) at (45:1cm) [vertex] {};
                \node (a3) at (-45:1cm) [vertex] {};
                \node (a4) at (-135:1cm) [vertex] {};
                \node (a5) at (135:1cm) [vertex] {};
             
                \draw[line width= 2, \colord] (a1) to (a2);
                \draw[line width= 2, \colord] (a1) to (a3);
                \draw[line width= 2, \colord] (a1) to (a4);        
                \draw[line width= 2, \colord] (a1) to (a5);

                \node (h4) at (0,-1.3) [labell] {$\phi(\sigma) = \binom{4}{2} = 6$};
                
            \end{tikzpicture}
        \end{subfigure}
        \caption{Potential on a $4$-star for some colorings.}
        \label{fig:pot}
\end{figure}

The algorithm consists of a random walk that starts at an arbitrary coloring and at each step moves to a neighboring coloring (differing on exactly one edge) uniformly at random. 
In the \emph{mild} version of the walk\footnote{A \emph{strict} and a \emph{weak} versions are also defined, but they are not discussed in here.} we have that $\sigma'$ is an \emph{out-neighbor} of $\sigma$ if $\psi(\sigma')\leq \psi(\sigma)$. In this sense, the chain can be viewed as a restricted version of the Glauber dynamics on (non-necessarily proper) edge-colorings, limited to movements that do not increase the potential. If the walk reaches a proper edge-coloring, then the algorithm stops and outputs it.

In~\cite{Dotan20}, based on simulations, the authors conjectured that the mild random walk in the $(2n-1)$-edge-colorings of $K_{2n}$ almost surely reaches a one-factorization (here $\Delta=k$), and that it does so in $\tilde{O}(n^4)$ steps.

Our main result studies the algorithm for $k\geq \Delta+1$, for \emph{all} graphs, providing an algorithmic version of Vizing's theorem.
\begin{theorem}\label{thm:random walk almost surely wins}
Let $G$ be any graph with maximum degree $\Delta$ and $k\geq \Delta+1$. Then the hill-climbing algorithm with the mild random walk starting at an arbitrary (non-proper) $k$-edge-coloring almost surely produces a proper $k$-edge-coloring of $G$.
\end{theorem}

To do so, we prove that for any arbitrary initial $k$-edge-coloring $\sigma$ there exists a proper $k$-edge-coloring $\sigma'$ and a sequence of single-edge recolorings from $\sigma$ to $\sigma'$ with non-increasing potential along the sequence. Thus the probability of getting "closer" to a proper edge-coloring is positive at each random single-edge recoloring. Since by design of the algorithm, we cannot get further from proper edge-colorings, the probability of reaching a proper edge-coloring tends to 1 as the number of steps tends to infinity, and the algorithm find a proper-coloring with probability $1$. 

While \Cref{thm:random walk almost surely wins} is best possible for Class II graphs, the case of Class I graphs and $k=\Delta$, which includes one-factorizations of regular graphs, remains as an open problem. 

Interestingly, Narboni  recently showed in \cite{Narboni2023} that any proper $k$-edge-coloring can be transformed into an optimal proper edge-coloring by a sequence of Kempe changes, even for Class I graphs. This can be seen as another algorithmic version of Vizing's theorem. Narboni's proof, as well as ours, is based on \emph{Vizing's fans} and the use of \emph{color-shift digraphs}. Though it should be noted that the setting is not the same: in Vizing's and Narboni's proofs, they start with an arbitrary proper edge-coloring with a lot of colors, and they do Kempe changes to reduce the number of colors, always keeping a proper edge-coloring; here, we start with a non-proper edge-coloring with the desired number of colors, and do single-edge recolorings to get closer and closer to a proper coloring until reaching it.

A basic analysis of our proof only shows that the expected running time is at most exponential in $n$, however this seems far from the truth. It would be interesting to provide a polynomial bound with a uniformly bounded exponent, or even to show that it is at most $n^{c}$, where $c=c(\Delta)$ is a constant only depending on the maximum degree.
The algorithm is conjectured to be a uniform sampler for proper $n$-edge-colorings of  $K_{2n}$, provided that it starts at a random $n$-edge-coloring~\cite{Dotan20}. Our analysis unfortunately does not provide any insight on the target distribution.

\section{Results}\label{sec:det}

In this section we show that, given any initial edge-coloring, a proper one can be reached in a polynomial number of steps, with (at most) the same number of colors. For the sake of conciseness, throughout this section we say \emph{colorings} instead of \emph{edge-colorings}. We need a few definitions between restating our result.

\smallskip

\noindent
\textit{Colorings.} Let $\sigma$ be a coloring of a graph $G$.
An \emph{$\alpha$-component} (or \emph{monochromatic component}) $H$ of $\sigma$ is a connected component of the graph $(V(G),E_\alpha)$ with $E_\alpha$ the set of edges colored $\alpha$ in $\sigma$. 
Note that a proper coloring is a coloring whose monochromatic components contain at most one edge.
An \emph{$(\alpha, \beta)$-component} (or \emph{bichromatic component}) $H$ of $\sigma$ is a connected component of $(V(G), E_{\alpha}\cup E_{\beta})$, i.e. induced by the edges colored $\alpha$ and $\beta$ in $\sigma$.
Let $\alpha$ be a color of $\sigma$ and $v$ a vertex. 
We recall that $d_{\sigma, \alpha}(v)$ denotes the number of edges incident with $v$ that are colored $
\alpha$ in $\sigma$. A pair of \emph{conflicting} edges is a pair of adjacent edges having the same color. 
We denote by $L_{\sigma}(v)$ the set of colors appearing on edges incident with $v$ in $\sigma$. 
We say a color $\alpha$ is \emph{missing around} $v$ in the coloring $\sigma$ if $\alpha \notin L_{\sigma}(v)$. In the figures, throughout the paper, a color inside a vertex represents a missing color around the vertex.

\smallskip

\noindent
\textit{Cherries.}
An $\alpha$-component containing exactly two adjacent edges is called an \emph{$\alpha$-cherry}. 
A cherry $c$ \emph{is centered at a vertex $v$} if $v$ is incident with both edges of $c$. 
In this case, we denote the cherry by $(u,v,w)$ with $u$ and $w$ the other endpoints of the edges of $c$.
A \emph{cherry coloring} is a coloring whose monochromatic components contain at most two edges, i.e., are cherries, single edges, or singletons. 

\smallskip

\noindent
\textit{Monotonic recolorings.} Given a coloring $\sigma$, a coloring $\sigma'$ is obtained from $\sigma$ by a \emph{monotonic recoloring} if $\sigma$ and $\sigma'$ differ on exactly one edge and $\psi(\sigma)\geq \psi(\sigma')$, where $\psi$ is the potential defined in~\eqref{eq:potential}, counting the number of conflicting pairs of edges. Equivalently, moving from $\sigma$ to $\sigma'$ is a valid transition of the mild random walk defined in the introduction. When we want to be more specific, we say that we \emph{monotonically recolor an edge $uv$ into a color $\alpha$}, meaning that this recoloring is monotonic.
We call a sequence of colorings $(\sigma_i)_{i\geq 0}$ \emph{monotonic} if $\sigma_{i+1}$ can be obtained from $\sigma_i$ by a monotonic recoloring, for all $i\geq 0$. The number of \emph{steps} is the length of the sequence.

\begin{theorem}\label{thm:det}
    Let $\sigma_0$ be a $k$-coloring of a graph $G$ such that $k \ge \Delta(G) + 1$. Then, there exists a monotonic recoloring sequence from $\sigma_0$ to a proper $k$-coloring of $G$ using $O(n^2\Delta)$ steps.
\end{theorem}

The rest of the section is devoted to proving the theorem. The strategy is as follows. If the current coloring contains a monochromatic component of size at least $3$, we prove that its potential can strictly decrease in a single step.
Otherwise, our coloring is a cherry coloring. 
If the cherry coloring is non-proper, then we can decrease the potential in $\Oh(n)$ steps.
Since the potential of any coloring is $\Oh(n\Delta^2)$ and the potential of a cherry coloring is $\Oh(n\Delta)$, we can obtain a proper coloring from any coloring using $\Oh(n\Delta^2 + n\cdot(n\Delta)) = \Oh(n^2\Delta)$ steps.

\subsection{Monochromatic components of size at least three}

First, we prove that we can always reduce the size of a monochromatic component of size at least $3$, while decreasing the potential. We start looking at the difference between two consecutive potentials.

\begin{lemma}\label{obs:potential counting}
Let $\sigma'$ be a coloring obtained from another coloring $\sigma$ by recoloring some edge $uv$ from color $\alpha$ to color $\beta$. Then 
$$\psi(\sigma)-\psi(\sigma')= d_{\sigma,\alpha}(u) + d_{\sigma,\alpha}(v) - (d_{\sigma,\beta}(u) + d_{\sigma,\beta}(v)) - 2$$
\end{lemma}

\begin{proof}
We can write
\begin{equation}\label{eq:potential_rewrite}
\psi(\sigma) = \sum_{v\in V(G)}\sum_{\alpha\in [k]} \binom{d_{\sigma,\alpha}(v)}{2} = \left(\frac{1}{2}\sum_{v\in V(G)}\sum_{\alpha\in [k]} d_{\sigma,\alpha}(v)^2\right) -|E(G)|. 
\end{equation}

From which we can derive 
\begin{align*}
2(\psi(\sigma) - \psi(\sigma')) &= d_{\sigma,\alpha}(u)^2 +  d_{\sigma,\alpha}(v)^2 +  d_{\sigma,\beta}(u)^2 +  d_{\sigma,\beta}(v)^2 -  d_{\sigma',\alpha}(u)^2  -  d_{\sigma',\alpha}(v)^2 - d_{\sigma',\beta}(u)^2  - d_{\sigma',\beta}(v)^2 \\
&=  d_{\sigma,\alpha}(u)^2 +  d_{\sigma,\alpha}(v)^2 +  d_{\sigma,\beta}(u)^2 +  d_{\sigma,\beta}(v)^2 \\
&\qquad -  (d_{\sigma,\alpha}(u) - 1)^2 -  (d_{\sigma,\alpha}(v) - 1)^2 - (d_{\sigma,\beta}(u) + 1)^2 - (d_{\sigma,\beta}(v) + 1)^2 \\ 
&= 2\left(d_{\sigma,\alpha}(u) + d_{\sigma,\alpha}(v) - (d_{\sigma,\beta}(u) + d_{\sigma,\beta}(v)) - 2\right) . \\
\end{align*}

\end{proof}

\begin{lemma}\label{lem:monosize2}
Let $\sigma$ be a $k$-coloring that induces a monochromatic component containing at least $3$ edges, with $k\geq \Delta+1$. Then, there exists a neighboring $k$-coloring $\sigma'$ such that $\psi(\sigma)>\psi(\sigma')$, i.e., we can decrease the potential in one step.
\end{lemma}

\begin{proof}

Let $\alpha$ be a color that induces a component containing at least $3$ edges. 
Then, this component contains either a vertex of degree at least $3$, a triangle or a path of length at least $3$.
In all three cases, there exists an edge $uv$ such that $\sigma(uv) = \alpha$, and $d_{\sigma,\alpha}(u) + d_{\sigma,\alpha}(v) \ge 4$; in other words there are at least four "incidences" between $\{u,v\}$ and the edges colored $\alpha$ by $\sigma$ (counting $uv$ twice). 

Since there are at most $2\Delta - 1 \le 2k - 3$ edges that are incident with $\{u,v\}$, there is a color $\beta$ appearing at most once among these edges around $u$ or $v$. In other words, $d_{\sigma,\beta}(u) + d_{\sigma,\beta}(v) \le 1$.

Let $\sigma'$ be the $k$-coloring obtained from $\sigma$ by recoloring $uv$ into $\beta$ (see \Cref{fig:3edgecomp}). 
Then using \Cref{obs:potential counting}:
\begin{align*}
\psi(\sigma) - \psi(\sigma') 
&= d_{\sigma,\alpha}(u) + d_{\sigma,\alpha}(v) - (d_{\sigma,\beta}(u) + d_{\sigma,\beta}(v)) - 2 \\ 
&\ge 4 - 1 - 2 \\
&\ge 1.
\end{align*}

In others words, intuitively, the decrease of the number of pairs of conflicting edges of color $\alpha$ in total in both vertices is greater than the increase of the number of pairs of conflicting edges of color $\beta$.
\end{proof}

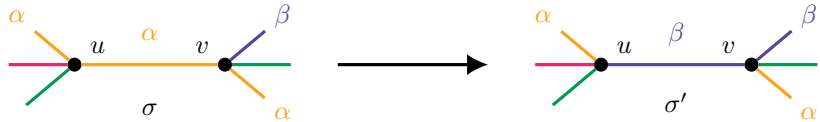
\begin{figure}[hbtp]
    \centering
        \tikzstyle{vertex}=[circle,draw, minimum size=7pt, scale=0.7, inner sep=1pt, fill = black]
        \tikzstyle{v}=[opacity=0,text opacity=1]
        
        \tikzstyle{v2}=[circle,draw, minimum size=12pt, scale=1.1, inner sep=2pt]
        \tikzstyle{vc}=[draw, minimum size=12pt, scale=1.1, inner sep=2pt]
        \tikzstyle{fleche}=[->,>=latex]
        \tikzstyle{garrow}=[arrows = {-Latex[width=8pt, length=8pt]}]
        \tikzstyle{labell}=[opacity=0,text opacity=1, scale =1]
        
        \begin{tikzpicture}[scale=1]

        \node (a1) at (0, 0) [vertex, label={20:$u$}] {};
        \node (a2) at (2,0) [vertex, label={160:$v$}] {};
        
        \node (a4) at (140:1cm) [v, \colord] {$\alpha$};
        \node (a5) at (220:1cm) [v, \colorf] {};
        \node (a6) at (180:1cm) [v, \colore] {};

        \draw[very thick, \colord] (a1) to (a2);  
        
        \node (h1) at (1,0.4)[labell, \colord] {$\alpha$};

        \node (h1) at (1,-0.6)[labell, scale=1] {$\sigma$};

        \draw[very thick, \colord] (a1) to (a4);
        \draw[very thick, \colorf] (a1) to (a5);
        \draw[very thick, \colore] (a1) to (a6);

        \tikzset{xshift=2cm}
        \node (a4) at (40:1cm) [v, \colorg] {$\beta$};
        \node (a5) at (-40:1cm) [v, \colord] {$\alpha$};
        \node (a6) at (0:1cm) [v, \colorf] {};

        \draw[very thick, \colorg] (a2) to (a4);
        \draw[very thick, \colord] (a2) to (a5);
        \draw[very thick, \colorf] (a2) to (a6);

        \draw[line width=1.2, garrow] (1.5, 0) to (3.5, 0);

        \tikzset{xshift=5cm}

        \node (a1) at (0, 0) [vertex, label={20:$u$}] {};
        \node (a2) at (2,0) [vertex, label={160:$v$}] {};
        
        \node (a4) at (140:1cm) [v, \colord] {$\alpha$};
        \node (a5) at (220:1cm) [v, \colorf] {};
        \node (a6) at (180:1cm) [v, \colore] {};

        \draw[very thick, \colorg] (a1) to (a2);  
        
        \node (h1) at (1,0.4)[labell, \colorg] {$\beta$};

        \node (h1) at (1,-0.5)[labell, scale=1] {$\sigma'$};

        \draw[very thick, \colord] (a1) to (a4);
        \draw[very thick, \colorf] (a1) to (a5);
        \draw[very thick, \colore] (a1) to (a6);

        \tikzset{xshift=2cm}
        \node (a4) at (40:1cm) [v, \colorg] {$\beta$};
        \node (a5) at (-40:1cm) [v, \colord] {$\alpha$};
        \node (a6) at (0:1cm) [v, \colorf] {};

        \draw[very thick, \colorg] (a2) to (a4);
        \draw[very thick, \colord] (a2) to (a5);
        \draw[very thick, \colorf] (a2) to (a6);

        \end{tikzpicture}
    \caption{On the left, the color $\alpha$ appears at least four times around $u$ and $v$ (counting $uv$ twice). By \Cref{lem:monosize2}, we can reduce the potential by recoloring $uv$ into a color appearing at most once around $u$ and $v$.}
        \label{fig:3edgecomp}
\end{figure}

Hence, we can apply \Cref{lem:monosize2} on our current coloring until we obtain a cherry coloring. 
Since we decrease the potential by at least one unit at each step and the potential is $\Oh(n\Delta^2)$, we can obtain a cherry coloring in $\Oh(n\Delta^2)$ steps.

\subsection{Bichromatic components}

It remains to show that we can decrease the potential of any cherry coloring that is not proper, in $O(n)$ steps.
The potential of a cherry coloring is linear in the number of cherries, thus, the potential decreases if and only if the number of cherries decreases.
We start by proving some properties on cherry colorings and their bichromatic components. In particular, we will show that a cherry (monochromatic path of length two) can be moved within a bichromatic component, allowing us to shift it in a convenient place for its eventual elimination.
Before showing it, we prove a number of simple but useful statements.

\begin{lemma}\label{prop:movecherry}
Let $\sigma$ be a cherry coloring and $uv$ be an edge of an $\alpha$-cherry of $\sigma$. If there is a color $\beta\notin L_{\sigma}(v)$ such that $u$ is not the center of a $\beta$-cherry, then we can monotonically recolor $uv$ into $\beta$.
\end{lemma}

\begin{proof}
Let $\sigma'$ be the $k$-coloring obtained from $\sigma$ by recoloring $uv$ into $\beta$ (see \Cref{fig:exchangecherry}). 
Since $uv$ is an edge of an $\alpha$-cherry, we have that $d_{\sigma, \alpha}(u) + d_{\sigma, \alpha}(v) = 3$. By hypothesis, $d_{\sigma, \beta}(u) + d_{\sigma, \beta}(v) \leq 1$. 
Using \Cref{obs:potential counting}, it follows that:
\begin{align*}
\psi(\sigma) - \psi(\sigma') 
&= (d_{\sigma,\alpha}(u) + d_{\sigma,\alpha}(v) - (d_{\sigma,\beta}(u) + d_{\sigma,\beta}(v)) - 2 \\ 
&\ge 3 - 1 - 2 \\
&\ge 0.
\end{align*}
\end{proof}

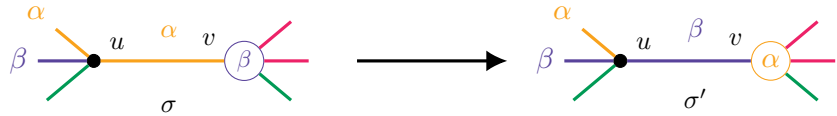
\begin{figure}[hbtp]
        \begin{center}
        \tikzstyle{vertex}=[circle,draw, minimum size=7pt, scale=0.7, inner sep=1pt, fill = black]
        \tikzstyle{v}=[opacity=0,text opacity=1, scale=1.1, minimum size=7pt]
        
        \tikzstyle{v2}=[circle,draw, minimum size=12pt, scale=1.1, inner sep=2pt]
        \tikzstyle{vc}=[draw, minimum size=12pt, scale=1.1, inner sep=2pt]
        \tikzstyle{fleche}=[->,>=latex]
        \tikzstyle{garrow}=[arrows = {-Latex[width=8pt, length=8pt]}]
        \tikzstyle{labell}=[opacity=0,text opacity=1, scale =1]
        
        \begin{tikzpicture}[scale=1]

        \node (a1) at (0, 0) [vertex, label={20:$u$}] {};
        \node (a2) at (2,0) [v2, \colorg, label={160:$v$}, scale=0.82] {$\beta$};
        
        \node (a4) at (140:1cm) [v, \colord] {$\alpha$};
        \node (a5) at (220:1cm) [v, \colorf] {};
        \node (a6) at (180:1cm) [v, \colorg] {$\beta$};

        \draw[very thick, \colord] (a1) to (a2);

        \node (h1) at (1,0.4)[labell, \colord] {$\alpha$};

        \node (h1) at (1,-0.6)[labell, scale=1] {$\sigma$};

        \draw[very thick, \colord] (a1) to (a4);
        \draw[very thick, \colorf] (a1) to (a5);
        \draw[very thick, \colorg] (a1) to (a6);

        \tikzset{xshift=2cm}
        \node (a4) at (40:1cm) [v, \colore] {};
        \node (a5) at (-40:1cm) [v, \colorf] {};
        \node (a6) at (0:1cm) [v, \colore] {};

        \draw[very thick, \colore] (a2) to (a4);
        \draw[very thick, \colorf] (a2) to (a5);
        \draw[very thick, \colore] (a2) to (a6);

        \draw[line width=1.2, garrow] (1.5, 0) to (3.5, 0);

        \tikzset{xshift=5cm}

        \node (a1) at (0, 0) [vertex, label={20:$u$}] {};
        \node (a2) at (2,0) [v2, \colord, label={160:$v$}] {$\alpha$};
        
        \node (a4) at (140:1cm) [v, \colord] {$\alpha$};
        \node (a5) at (220:1cm) [v, \colorf] {};
        \node (a6) at (180:1cm) [v, \colorg] {$\beta$};

        \draw[very thick, \colorg] (a1) to (a2);  
        
        \node (h1) at (1,0.4)[labell, \colorg] {$\beta$};

        \node (h1) at (1,-0.5)[labell, scale=1] {$\sigma'$};
        \draw[very thick, \colord] (a1) to (a4);
        \draw[very thick, \colorf] (a1) to (a5);
        \draw[very thick, \colorg] (a1) to (a6);

        \tikzset{xshift=2cm}
        \node (a4) at (40:1cm) [v, \colore] {};
        \node (a5) at (-40:1cm) [v, \colorf] {};
        \node (a6) at (0:1cm) [v, \colore] {};

        \draw[very thick, \colore] (a2) to (a4);
        \draw[very thick, \colorf] (a2) to (a5);
        \draw[very thick, \colore] (a2) to (a6);

        \end{tikzpicture}
        \end{center}
    \caption{Illustration of \Cref{prop:movecherry}. We monotonically recolor the edge $uv$ of the $\alpha$-cherry (centered at $u$ here). }
        \label{fig:exchangecherry}
\end{figure}

\begin{lemma}\label{prop:numbermissing}
Let $\sigma$ be a cherry coloring, with $k\geq \Delta(G)+1$ colors. The number of colors missing around a vertex $v$ in $\sigma$ is at least $m+1$, where $m$ is the number of cherries of $\sigma$ centered at $v$.
\end{lemma}

\begin{proof}
Suppose that $v$ is the center of exactly $m$ cherries, then $m$ colors appear twice around $v$ and the other colors appear at most once. Thus, we have that $ | L_{\sigma}(v) | =  m + (d(v) - 2m) = d(v) - m \leq  k - (m+1)$. Hence, there are at least $m+1$ colors missing around $v$.
\end{proof}

Next result allows us to either break or change the color of a given cherry, this will prove most useful in the coming subsections.

\begin{lemma}\label{lem:exchangecherry}
Let $\sigma$ be a cherry coloring, $uv$ an edge of an $\alpha$-cherry of $\sigma$ and $\gamma \in L_{\sigma}(v)$. Then, we can monotonically recolor $uv$ into a color $\beta\notin L_{\sigma}(u)\cap L_{\sigma}(v)$, with $\beta\neq \alpha, \gamma$.
\end{lemma}
\begin{proof}
If $v$ is the center of exactly $m$ cherries, by \Cref{prop:numbermissing}, $v$ has at least  $m+1$ missing colors. 

If there is $\beta\notin L_{\sigma}(v)$ such that $u$ is not the center of a $\beta$-cherry, then, by \Cref{prop:movecherry}, we can monotonically recolor $uv$ into $\beta$. Since $\alpha,\gamma\in L_{\sigma}(v)$, $\beta$ is different from them.

Otherwise, for any $\zeta \notin L_{\sigma}(v)$, $u$ is the center of a $\zeta$-cherry. Thus, there are at least $m+1$ cherries around $u$, and, by \Cref{prop:numbermissing}, it has at least $m+2$ missing colors. 
Since $v$ is the center of $m$ cherries, at least two of the missing colors around $u$ do not form a cherry centered at $v$.
In particular, among these two colors, one is different from $\gamma$, we fix this to be $\beta$. 
Since $\beta \notin L_{\sigma}(u)$, we have that $\beta \neq \alpha$. Hence, we can monotonically recolor $uv$ into $\beta$ by \Cref{prop:movecherry}.
\end{proof}

Using the previous lemmas, in the following two results we show that, in a cherry coloring, if a bichromatic component containing a cherry has minimum degree one, we can decrease its potential.

\begin{lemma}\label{cl:pathbi}
Let $\sigma$ be a coloring, $c = (u,v,w)$ be a $\gamma_1$-cherry, let $\gamma_2\notin L_{\sigma}(v)$ and $H$ the $(\gamma_1, \gamma_2)$-component of $\sigma$ containing $c$. If $H$ contains a path $P = (x_1, \dots, x_\ell = v)$ such that $x_2, \dots, x_\ell$ have degree exactly $2$ in $H$, then we can monotonically recolor the edges of $P$ such that either we strictly reduce the potential of the initial coloring $\sigma$ or we obtain a $\gamma_i$-cherry centered at $x_2$, for some $i\in \{1,2\}$, in at most $\ell$ steps. 
\end{lemma}

\begin{proof}
We will prove it by induction on $\ell \geq 2$.  If $\ell=2$, then $\sigma$ satisfies the properties of the desired coloring. 
Suppose that $\ell\geq 3$ and note that $x_{\ell-1}$ is either $u$ or $w$, by symmetry, assume it is $u$.
If $\sigma$ contains a monochromatic component of size at least 3, we can decrease the potential with one recoloring using \Cref{lem:monosize2}. So we may assume that $\sigma$ is a cherry coloring and 
$\sigma(x_{\ell-2} u) = \gamma_{2}$. 
By \Cref{prop:movecherry}, we can monotonically recolor $uv$ into $\gamma_{2}$. 
Now, we have a $\gamma_2$-cherry $c'=(x_{\ell-2}, u,v)$ with $\gamma_1\notin L_{\sigma}(u)$ and the bichromatic component contains $P'=(x_1, \dots, x_{\ell-1} = u)$ with $x_2\dots, x_{\ell-1}$ of degree $2$. 
By induction, we can monotonically recolor edges to either decrease the potential or obtain a $\gamma_i$-cherry centered at $x_2$ for some $i\in\{1,2\}$, in at most $\ell-1$ steps. The desired result follows.
\end{proof}

\begin{lemma}\label{lem:bichromaticcycle}
Let $\sigma$ be a cherry coloring and $c$ be a $\gamma_1$-cherry such that $\gamma_2$ is missing around the center of $c$ and the $(\gamma_1, \gamma_2)$-component containing $c$ has a vertex of degree $1$.
Then, we can decrease the potential in $O(n)$ steps.
\end{lemma}

\begin{proof}
Let $H$ be the $(\gamma_1, \gamma_2)$-component containing $c$, $w$ a vertex of degree $1$ in it, and $v$ the center of $c$.
\Cref{fig:movingbi} illustrates this proof. 
Consider a shortest path $P$ in $H$ from $v$ to $w$. 
For each cherry $c'=(x,y,z)$ in $H$ with $y\in V(P)$ but $z\notin V(P)$, we use \Cref{lem:exchangecherry} to recolor $yz$ into a color different from $\gamma_1,\gamma_2$; this removes the cherry $c'$ without affecting $P$ or adding any new cherry in colors $\gamma_1$ or $\gamma_2$. After that, let $y$ be the vertex of degree $3$ in $P$ that is closest to $v$ (if it exists). Then $y$ is the center of a cherry $c'=(x,y,z)$ contained in $P$. Say that $x$ is the vertex closer to $v$ in $c$. Using \Cref{lem:exchangecherry} we can monotonically recolor $xy$ into a color different from $\gamma_1,\gamma_2$. After these steps, we obtain a coloring $\sigma'$ containing $c$ such that the $(\gamma_1, \gamma_2)$-component $H'$ containing $c$ also contains a path $P'=(x_1=x,x_2,\dots, x_\ell=v)$ where $x_2,\dots,x_\ell$ have degree $2$ and $x_1$ has degree $1$.

By \Cref{cl:pathbi}, we can monotonically recolor edges of $P'$ such that we either decrease the potential of the coloring $\sigma'$ or obtain a coloring $\sigma''$ with a $\gamma_i$-cherry $c''$ centered at $x_2$ in at most $\ell$ steps, for $i\in \{1,2\}$. 
In the first case, we can conclude.
Otherwise, observe that the $(\gamma_1,\gamma_2)$-component containing $c''$ is still $H'$.
Since $x_1$ has degree $1$ and $x_2$ has degree $2$ in $H'$, we have $\gamma_j\notin L_{\sigma}({x_1})\cup L_{\sigma}({x_2})$, for $j\neq i$. Recoloring $x_1x_2$ into $\gamma_{j}$ decreases the potential of $\sigma''$, concluding the proof. 

As there are at most $2n$ cherries in $H$, in the process of decreasing the potential, we have done at most $ 2n+ |P'|= O(n)$ recolorings.
\end{proof}

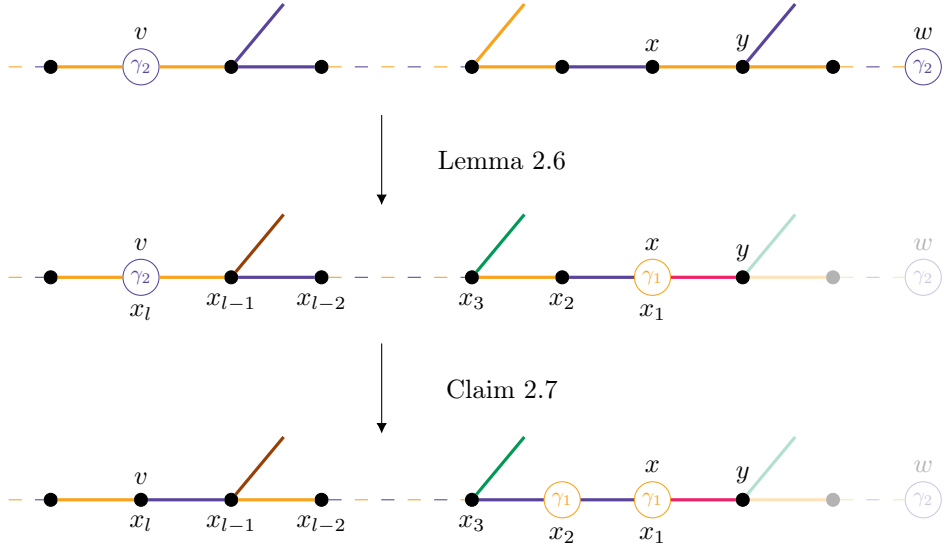
\begin{figure}[hbtp]
        \begin{center}
        \tikzstyle{vertex}=[circle,draw, minimum size=7pt, scale=0.7, inner sep=1pt, fill = black]
        \tikzstyle{v}=[circle,draw, minimum size=12pt, scale=0.8, inner sep=2pt]
        \tikzstyle{fleche}=[->,>=latex]
        \tikzstyle{garrow}=[arrows = {-Latex[width=4pt, length=4pt]}]
        \tikzstyle{labell}=[opacity=0,text opacity=1, scale =1]
        \begin{tikzpicture}[scale=0.8]

                \node (a1) at (0, 0) [v, label={90:$v$}, \colorg] {$\gamma_2$};
                \node (a2) at (1.5, 0) [vertex]{};  
                \node (a3) at (-1.5, 0) [vertex] {};
                \node (a4) at (3, 0) [vertex] {};
                \node (a5) at (5.5, 0) [vertex] {};
                \node (a6) at (7, 0) [vertex] {};
                \node (a7) at (8.5, 0) [vertex, label={90:$x$}] {};

                \node (b1) at (2.5,1.2) [labell] {};
                \node (b2) at (6.5,1.2) [labell] {};

                \draw[very thick, \colord] (a1) to (a2);
                \draw[very thick, \colord] (a1) to (a3);
                \draw[very thick, \colorg] (a2) to (a4);

                \draw[dashone, draw=\colord, postaction={draw=\colorg, dash phase= \dash}] (a4) to (a5);
                
                \draw[very thick, \colord] (a5) to (a6);
                \draw[very thick, \colorg] (a6) to (a7);
                \draw[dashone, draw=\colord, postaction={draw=\colorg, dash phase= \dash}] (-2.2,0) to (a3);

                \draw[very thick, \colorg] (a2) to (b1);
                \draw[very thick, \colord] (a5) to (b2);

                \node (a8) at (10, 0) [vertex, label={90:$y$}] {};
                \node (a9) at (11.5, 0) [vertex] {};
                \node (a0) at (13, 0) [v, \colorg, label={90:$w$}] {$\gamma_2$};
                
                \node (b3) at (11,1.2) [labell] {};
                \draw[very thick, \colord] (a7) to (a8);
                \draw[very thick, \colord] (a8) to (a9);
                
                \draw[very thick, \colorg] (a8) to (b3);
                \draw[dashone, draw=\colord, postaction={draw=\colorg, dash phase= \dash}] (a9) to (a0);

                \draw[garrow] (4, -0.8) to (4, -2.3);
                
                \node (b2) at (6,-1.55) [labell] {Lemma~\ref{lem:exchangecherry}};
                
            \tikzset{yshift=-3.5cm};
            
               \node (a1) at (0, 0) [v, \colorg , label={-90:$x_{l}$}, label={90:$v$}] {$\gamma_2$};
                \node (a2) at (1.5, 0) [vertex, label={-90:$x_{l-1}$}]{};  
                \node (a3) at (-1.5, 0) [vertex] {};
                \node (a4) at (3, 0) [vertex, label={-90:$x_{l-2}$}] {};
                \node (a5) at (5.5, 0) [vertex, label={-90:$x_3$}] {};
                \node (a6) at (7, 0) [vertex, label={-90:$x_2$}] {};
                \node (a7) at (8.5, 0) [v, \colord, label={-90:$x_1$}, label={90:$x$}] {$\gamma_1$};

                \node (b1) at (2.5,1.2) [labell] {};
                \node (b2) at (6.5,1.2) [labell] {};

                \draw[very thick, \colord] (a1) to (a2);
                \draw[very thick, \colord] (a1) to (a3);
                \draw[very thick, \colorg] (a2) to (a4);

                \draw[dashone, draw=\colord, postaction={draw=\colorg, dash phase= \dash}] (a4) to (a5);
                
                \draw[very thick, \colord] (a5) to (a6);
                \draw[very thick, \colorg] (a6) to (a7);
                
                \draw[dashone, draw=\colord, postaction={draw=\colorg, dash phase= \dash}] (-2.2,0) to (a3);

                \draw[very thick, \colorc] (a2) to (b1);
                \draw[very thick, \colorf] (a5) to (b2);

                \node (a8) at (10, 0) [vertex, label={90:$y$}] {};

                \draw[very thick, \colore] (a7) to (a8);

            \begin{scope}[transparency group, opacity=0.3]

                \node (a9) at (11.5, 0) [vertex] {};
                \node (a0) at (13, 0) [v, \colorg, label={90:$w$}] {$\gamma_2$};
                
                \node (b3) at (11,1.2) [labell] {};
                \draw[very thick, \colord] (a8) to (a9);
                
                \draw[very thick, \colorf] (a8) to (b3);
                \draw[dashone, draw=\colord, postaction={draw=\colorg, dash phase= \dash}] (a9) to (a0);

            \end{scope}
                
                \draw[garrow] (4, -1.1) to (4, -2.6);
                
                \node (b2) at (6,-1.85) [labell] {Claim~\ref{cl:pathbi}};
            
            \tikzset{yshift=-3.7cm};
            
               \node (a1) at (0, 0) [vertex,label={-90:$x_{l}$}, label={90:$v$}] {};
                \node (a2) at (1.5, 0) [vertex, label={-90:$x_{l-1}$}]{};  
                \node (a3) at (-1.5, 0) [vertex] {};
                \node (a4) at (3, 0) [vertex, label={-90:$x_{l-2}$}] {};
                \node (a5) at (5.5, 0) [vertex, label={-90:$x_3$}] {};
                \node (a6) at (7, 0) [v, label={-90:$x_2$}, \colord] {$\gamma_1$};
                \node (a7) at (8.5, 0) [v, \colord, label={-90:$x_1$}, label={90:$x$}] {$\gamma_1$};

                \node (b1) at (2.5,1.2) [labell] {};
                \node (b2) at (6.5,1.2) [labell] {};

                \draw[very thick, \colorg] (a1) to (a2);
                \draw[very thick, \colord] (a1) to (a3);
                \draw[very thick, \colord] (a2) to (a4);

                \draw[dashone, draw=\colorg, postaction={draw=\colord, dash phase= \dash}] (a4) to (a5);
                
                \draw[very thick, \colorg] (a5) to (a6);
                \draw[very thick, \colorg] (a6) to (a7);
                \draw[dashone, draw=\colord, postaction={draw=\colorg, dash phase= \dash}] (-2.2,0) to (a3);

                \draw[very thick, \colorc] (a2) to (b1);
                \draw[very thick, \colorf] (a5) to (b2);
                
                \node (a8) at (10, 0) [vertex, label={90:$y$}] {};

                \draw[very thick, \colore] (a7) to (a8);

            \begin{scope}[transparency group, opacity=0.3]

                \node (a9) at (11.5, 0) [vertex] {};
                \node (a0) at (13, 0) [v, \colorg, label={90:$w$}] {$\gamma_2$};
                
                \node (b3) at (11,1.2) [labell] {};
                \draw[very thick, \colord] (a8) to (a9);
                
                \draw[very thick, \colorf] (a8) to (b3);
                \draw[dashone, draw=\colord, postaction={draw=\colorg, dash phase= \dash}] (a9) to (a0);

            \end{scope}

        \end{tikzpicture}
        \end{center}
        \caption{Illustration of \Cref{lem:bichromaticcycle}. A $\gamma_1$-cherry $c$ (orange here) whose center $v$ misses the color $\gamma_2$ (violet here). The $(\gamma_1, \gamma_2)$-component $H$ containing $c$ has a vertex $w$ of degree one. First, we apply \Cref{lem:exchangecherry} to construct in $H$ a path $P$ of vertices of degree at most $2$ from $v$ to a vertex $x$ of degree $1$. Then, we apply \Cref{cl:pathbi} on $P$ to "move" the center of the cherry from $v$ to $x_2$. Finally, we can just recolor $x_1x_2$ to decrease the potential.}
        \label{fig:movingbi}
\end{figure}

As  a direct consequence we get the following.
\begin{corollary}\label{coro:bichromaticcycle}
Let $\sigma$ be a cherry $k$-coloring and let $H$ be a bichromatic component of $\sigma$. If $c,c'$ are cherries in $H$ whose centers have degree $2$, then we can decrease the potential in $O(n)$ steps.
\end{corollary}
\begin{proof}
Let $c=(u,v,w)$ and $c'=(x,y,z)$, and assume that, in $H$, the distance from $v$ to $z$ is at least the distance from $v$ to $x$. By \Cref{lem:exchangecherry}, monotonically recolor $yz$ into a color different from the ones in $H$. Now $y$ has degree $1$ in the bichromatic component containing $c$, and we can conclude by \Cref{lem:bichromaticcycle}.
\end{proof}

\subsection{Color-shift digraph}

Using the previous properties, we are ready to show how one can decrease the potential of a non-proper cherry coloring using monotonic recolorings. 
Our algorithm is based on the analysis of an auxiliary digraph that indicates the different options we have to monotonically recolor the edges incident with a vertex. 
The \emph{color-shift digraph $\cD_\sigma(v)$ of $\sigma$ around $v$} is a digraph whose vertices are colors in $[k]$ and there is an arc $(\alpha,\beta)$ if there is $uv\in E(G)$ such that $\sigma(uv) = \alpha$ and $\beta\notin L_{\sigma}(u)$ (see \Cref{fig:colshift} for an example). We say that a color $\alpha$ is \emph{marked} in $\cD_{\sigma}(v)$ if there is an $\alpha$-cherry centered at $v$.
Since there is always a color missing at any $u$, $\alpha$ has out-degree $0$ in $\cD_\sigma(v)$ if and only if $\alpha\notin L_{\sigma}(v)$. Also, $\cD_\sigma(v)$ has at most $d(v)$ colors with outdegree at least one so a path between two colors or a cycle (if it exists) has length at most $d(v)$.

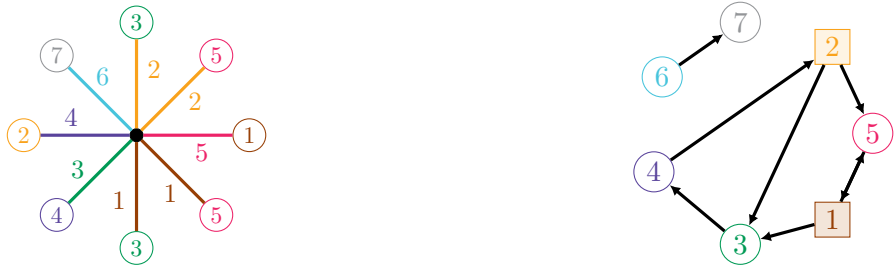
\begin{figure}[h]
       
        \tikzstyle{vertex}=[circle,draw, minimum size=7pt, scale=0.7, inner sep=1pt, fill = black]
        \tikzstyle{v}=[circle,draw, minimum size=12pt, scale=0.9, inner sep=2pt]
        
        \tikzstyle{v2}=[circle,draw, minimum size=12pt, scale=1.1, inner sep=2pt]
        \tikzstyle{vc}=[draw, minimum size=12pt, scale=1.1, inner sep=2pt]
        \tikzstyle{fleche}=[->,>=latex]
        \tikzstyle{garrow}=[arrows = {-Latex[width=4pt, length=4pt]}]
        \tikzstyle{labell}=[opacity=0,text opacity=1, scale =1]
        
        \begin{subfigure}[t]{.5\textwidth}
            \centering
            \begin{tikzpicture}[scale=1]
    
            \node (a0) at (0,0) [vertex] {};
    
            \node (a4) at (-45:1.5cm) [v, \colore] {$5$};
            \node (a5) at (-90:1.5cm) [v, \colorf] {$3$};
            \node (a1) at (45:1.5cm) [v, \colore] {$5$};
            \node (a2) at (90:1.5cm) [v, \colorf] {$3$};
            \node (a3) at (0:1.5cm) [v, \colorc] {$1$};
            \node (a6) at (135:1.5cm) [v, \colora] {$7$};
            \node (a8) at (180:1.5cm) [v, \colord] {$2$};
            \node (a9) at (-135:1.5cm) [v, \colorg] {$4$};

            \draw[very thick, \colorc] (a0) to (a4);
            \draw[very thick, \colorc] (a0) to (a5);
            \draw[very thick, \colord] (a0) to (a1);
            \draw[very thick, \colord] (a0) to (a2);
            \draw[very thick, \colore] (a0) to (a3);
            \draw[very thick, \colorb] (a0) to (a6);
            \draw[very thick, \colorg] (a0) to (a8);
            \draw[very thick, \colorf] (a0) to (a9);

            \node (h1) at (-60:.9cm) [labell, \colorc] {$1$};
            \node (h2) at (-105:.9cm) [labell, \colorc] {$1$};
            \node (h1) at (30:.9cm) [labell, \colord] {$2$};
            \node (h2) at (75:.9cm) [labell, \colord] {$2$};
            \node (h3) at (-15:.9cm) [labell, \colore] {$5$};
            \node (h6) at (120:.9cm) [labell,\colorb] {$6$};
            \node (h8) at (165:.9cm) [labell, \colorg] {$4$};
            \node (h9) at (-150:.9cm) [labell, \colorf] {$3$};
            
            \end{tikzpicture}
        \end{subfigure}
        \hfill
        \begin{subfigure}[t]{.5\textwidth}
            \centering
            \begin{tikzpicture}[scale=1]
    
               \node (a1) at (-50:1.5cm) [vc, \colorc, fill=\colorc!10] {$1$};
                \node (a3) at (-100:1.5cm) [v2, \colorf] {$3$};
                \node (a5) at (0:1.5cm) [v2, \colore] {$5$};
                \node (a2) at (50:1.5cm) [vc, \colord, fill=\colord!10] {$2$};
                \node (a8) at (100:1.5cm) [v2, \colora] {$7$};
                \node (a6) at (150:1.5cm) [v2, \colorb] {$6$};
                \node (a4) at (200:1.5cm) [v2, \colorg] {$4$};

                \draw[very thick, garrow] (a1) to (a5);
                \draw[very thick, garrow] (a1) to (a3);
                \draw[very thick, garrow] (a2) to (a5);
                \draw[very thick, garrow] (a2) to (a3);
                \draw[very thick, garrow] (a3) to (a4);
                \draw[very thick, garrow] (a4) to (a2);
                \draw[very thick, garrow] (a6) to (a8);
                \draw[very thick, garrow] (a5) to (a1);
            \end{tikzpicture}
        \end{subfigure}
        
        \caption{An example of a coloring $\sigma$ (on the left) around a vertex $v$, and the color-shift digraph of $\sigma$ around $v$ (on the right) where marked vertices are represented as squares. We do not represent isolated vertices in $\cD_\sigma(v)$. }
        \label{fig:colshift}
\end{figure}

\begin{lemma}\label{lem:outdegree0}
Let $\sigma$ be a cherry $k$-coloring and $v\in V(G)$. If there are two colors $\alpha$ and $\beta$ in $\cD_\sigma(v)$ such that $\beta$ is marked, $\alpha$ has out-degree $0$, and there is a path from $\beta$ to $\alpha$, then we can decrease the potential in at most $d(v)$ monotonic steps.
\end{lemma}

\begin{proof}
Let $\beta$ be a marked color that minimises the distance to $\alpha$ in $\cD_\sigma(v)$. Fix a shortest path from $\beta$ to $\alpha$, let $\ell$ be its length, and $\gamma$ the successor of $\beta$ in it. As $(\beta,\gamma)$ is an arc of $\cD_{\sigma}(v)$, there is $u\in V(G)$ with $\sigma(uv)=\beta$ and $\gamma \notin L_{\sigma}(u)$. 

We prove the statement of the lemma by induction on $\ell \geq 1$. If $\ell=1$, then $\gamma=\alpha$ and since  $\alpha\notin L_{\sigma}(u)\cup L_{\sigma}(v)$, we can recolor $uv$ with $\alpha$, reducing the number of cherries, and thus the potential.
If $\ell\geq 2$, then $\gamma\in L_{\sigma}(v)$, but by minimality of $\beta$, $\gamma$ is not marked. Since $\gamma\notin L_{\sigma}(u)$, by \Cref{prop:movecherry}, we can monotonically recolor $uv$ with $\gamma$, moving the mark from $\beta$ to $\gamma$. Observe that the path between $\gamma$ and $\alpha$ in $\cD_\sigma(v)$ is preserved by this recoloring, as $\beta$ does not appear in it.
We now have a coloring $\sigma'$ such that,  in $\cD_{\sigma'}(v)$, $\gamma$ is marked, $\alpha$ has outdegree zero and there is a path of length $\ell-1$ between $\gamma$ and $\alpha$, so we conclude by induction.
\end{proof}

\begin{lemma}\label{lem:nocircuit}
Let $\sigma$ be a cherry $k$-coloring and $v \in V(G)$ a center of a $\beta$-cherry.
If  $\cD_{\sigma}(v)$ does not contain a cycle with a marked color, then we can monotonically recolor the graph to a coloring $\sigma'$ such that $\cD_{\sigma'}(v)$ contains a cycle with a marked color in at most $d(v)$ steps.
\end{lemma}
\begin{proof}
Let $P$ be a longest path from $\beta$ in $\cD_{\sigma}(v)$.
By \Cref{lem:outdegree0}, $P$ has no vertex of degree $0$.
By hypothesis, $P$ does not form a cycle. 
Hence, $P$ can be separated to a path $P_0$ ending on a color $\alpha$ and a cycle $\C$ containing $\alpha$.
Let $\gamma$ be the last marked color of $P_0$ and $\ell$ the distance from $\gamma$ to $\alpha$ along $P_0$. 

As in \Cref{lem:outdegree0}, we prove the statement of the lemma by induction on $\ell \ge 1$ by "moving" the cherry along $P_0$ until it reaches $\alpha$. 
Hence, we obtain a coloring $\sigma'$ such that $\alpha$ is a marked color in the cycle $\C$ in $\cD_{\sigma'}(v)$ in $d(v)$ steps.
\end{proof}

A \emph{chord} of a cycle $C$ in a digraph $\cD$ is an arc from a vertex $u$ of $\C$ to a vertex of $\C$ which is not the succesor of $u$ in $C$.  
A cycle $\C$ with no chords is called \emph{chordless}.

\begin{lemma}\label{lem:circuitothers}
Let $\sigma$ be a cherry $k$-coloring and $v\in V(G)$. Suppose that $\C$ is a chordless cycle in $\cD_{\sigma}(v)$ with at least one marked color $\beta$. Then we can monotonically recolor $\sigma$ into a coloring $\sigma'$ such that $\C$ is a cycle of $\cD_{\sigma'}(v)$ and $\beta$ is the only marked color in it. This can be done in at most $d(v)$ steps.
\end{lemma}

\begin{proof}
Suppose $\alpha \neq \beta$ is a marked color in $\C$ and let $c = (u,v,w)$ be the $\alpha$-cherry around $v$. 
Let $\gamma$ be the successor of $\alpha$ in $\C$.
We have that $\gamma \in L_{\sigma}(v)$ and we may assume, by symmetry, that $\gamma\notin L_{\sigma}(u)$. We will not recolor $uv$ to ensure that $(\alpha,\gamma)$ is an arc of the color-shift digraph of the final coloring.

By \Cref{lem:exchangecherry}, we can monotonically recolor $vw$ with a color $\delta\neq \alpha,\gamma$ such that $\delta\notin L_{\sigma}(v)\cap L_{\sigma}(w)$. This deletes $c$, but may create a new $\delta$-cherry.
Let us argue that $\delta\notin V(\C)$. If $\delta\notin L_{\sigma}(v)$, then $\delta$ has out-degree $0$ in $\cD_{\sigma}(v)$ and cannot be in $\C$.  If $\delta\notin L_{\sigma}(w)$, then $(\alpha,\delta)$ is an arc of $\cD_{\sigma}(v)$, but since $\delta$ is not $\gamma$, the color is not in $\C$ (otherwise, $(\alpha, \delta)$ would be a chord of $C$).

Therefore, $\C$ is preserved in the new coloring, but the number of marks in it has decreased, as $\alpha$ is no longer marked and $\delta$ is not in $\C$. We may repeat this procedure until $\beta$ is the only marked color in $\C$.
\end{proof}

\begin{lemma}\label{lem:circuitone}
Let $\sigma$ be a cherry $k$-coloring and $v\in V(G)$.
If there exists a cycle in $\cD_\sigma(v)$ with exactly one marked color, then we can decrease the potential in $\Oh(n)$ monotonic steps. 
\end{lemma}

\begin{proof} \Cref{fig:movearound} illustrates this proof.
Let $\C$ be the cycle, and $\beta$ the marked color.
Let $c = (u,v,w)$ be the $\beta$-cherry centered at $v$. 
Since $k\geq \Delta+1$, there is a color $\alpha$ missing around $v$.
If the $(\beta, \alpha)$-component $H$ containing $c$ has a vertex of degree $1$, we can conclude by \Cref{lem:bichromaticcycle}.
Hence, assume that its minimum degree is at least $2$. 

We will ``move'' the cherry centered at $v$ around $\C$. Let $\gamma$ be the successor of $\beta$ in $\C$.
We may assume, by symmetry, that $\gamma\notin L_{\sigma}(w)$. As there is no $\gamma$-cherry around $v$, we can monotonically recolor $vw$ with $\gamma$ to obtain a new coloring such that $\gamma$ is the only marked color in $\C$; this can be seen as passing the mark from $\beta$ to $\gamma$.
We repeat this process at most $d(v) = \Oh(n)$ times until the first time we obtain a $\beta$-cherry again; this will happen as $\C$ is a cycle. Let $\sigma'$ be the new coloring and $c' = (u,v,x)$ the new $\beta$-cherry. 
By construction of $\cD_{\sigma}(v)$, we have that $\beta \notin L_\sigma(x)$, hence $x$ is different from $w$ since $\sigma(uw) = \beta$.
As we only recolor edges incident with $v$, the $(\alpha, \beta)$-component $H'$ of $\sigma'$ containing $c'$ differs from $H$ in only two edges: $vw$ has been replaced by $vx$. In particular, $w$ still belongs to the same $(\alpha, \beta)$-component as $u$, since $w$ has degree at least two in the connected component $H$ and we removed only one of its incident edges. In $H'$, either $w$ has degree $1$, or it has degree $2$ and is the center of an $\alpha$-cherry. We can use \Cref{lem:bichromaticcycle} and \Cref{coro:bichromaticcycle} respectively to conclude.
\end{proof}

\begin{proof}[Proof of \Cref{thm:det}.]
First we repeatedly use \Cref{lem:monosize2} to get a cherry coloring in $O(n\Delta^2)$ steps. This cherry coloring as potential $O(n\Delta)$.
Then using the four previous lemmas, there exists a monotonic reconfiguration sequence from any cherry coloring $\sigma$ that decreases the potential. Indeed, consider a marked color $\beta$ in $\cD_\sigma(v)$. 
If there is a path from $\beta$ to a color of out-degree $0$, we can decrease the potential in at most $d(v)$ steps by \Cref{lem:outdegree0}.
Otherwise, by \Cref{lem:nocircuit}, we can recolor the graph such that $\cD_\sigma(v)$ contains a cycle with a marked color in at most $d(v)$ steps.
Denote by $\C$ a chordless cycle that has at least one marked vertex.
By \Cref{lem:circuitothers}, we can remove the marks of all the other colors in $\C$ and then, by \Cref{lem:circuitone}, decrease the potential in $\Oh(n)$ steps by unit decrease. This concludes the proof.
\end{proof}

\begin{figure}[h]
       
        \tikzstyle{vertex}=[circle,draw, minimum size=7pt, scale=0.7, inner sep=1pt, fill = black]
        \tikzstyle{v}=[circle,draw, minimum size=12pt, scale=0.9, inner sep=2pt]
        
        \tikzstyle{v2}=[circle,draw, minimum size=12pt, scale=1.1, inner sep=2pt]
        \tikzstyle{vc}=[draw, minimum size=12pt, scale=1.2, inner sep=3pt]
        \tikzstyle{fleche}=[->,>=latex]
        \tikzstyle{garrow}=[arrows = {-Latex[width=4pt, length=4pt]}]
        
        \tikzstyle{garrow2}=[arrows = {-Latex[width=6pt, length=6pt]}]
        
        \tikzstyle{garrow3}=[arrows = {-Latex[width=10pt, length=10pt]}]
        \tikzstyle{labell}=[opacity=0,text opacity=1, scale =1]
        
        \centering
            \begin{tikzpicture}[scale=1.1]
    
            \node (a0) at (0,0) [v, \colore, label={160:$v$}, scale=1.15] {$\alpha$};
    
            \node (a4) at (-45:1.5cm) [v, \colorf] {$\zeta$};
            \node (a5) at (-90:1.5cm) [v, \colord, label={180:$w$}, scale=1.15] {$\gamma$};
            \node (a3) at (0:1.5cm) [v, \colorg, label={20:$x$}] {$\beta$};
            \node (a9) at (-135:1.5cm) [vertex, label={180:$u$}] {};

            \node (b1) at (-100:2.5cm) [vertex] {};
            \node (b2) at (-125:2.5cm) [vertex] {};
            
            \draw[line width = 1.9, \colore] (a5) to (b1);
                \draw[dashone, line width = 1.9, draw=\colorg, postaction={draw=\colore, dash phase= \dash}] (b1) to[bend left=20] (b2);
                
            \draw[line width = 1.9, \colore] (b2) to (a9);

            \draw[very thick, dashed, opacity=.6] (a0) to (180:.6cm);
            \draw[very thick, \colord] (a0) to (a4);
            \draw[line width = 1.9, \colorg] (a0) to (a5);
            \draw[very thick, \colorf] (a0) to (a3);
            \draw[line width = 1.9, \colorg] (a0) to (a9);

            \draw[line width=1.8, garrow2, opacity=1, \colorc] (-80:2cm) arc (-80:-10:2cm);

            \node (h1) at (-60:.9cm) [labell, \colord] {$\gamma$};
            \node (h2) at (-105:.9cm) [labell, \colorg] {$\beta$};
            \node (h3) at (-15:.9cm) [labell, \colorf] {$\zeta$};
            \node (h9) at (-150:.9cm) [labell, \colorg] {$\beta$};

            \node (h0) at (-1,-3.1) [labell, scale=1] {$\sigma$};
            
            \node (h0) at (2.5,-2) [labell, scale=1] {$\cD_{\sigma}(v)$};

            \node (i0) at (-1.5,-5) [labell] {};
            \node (i2) at (3,0) [labell] {};
            \draw[thick, opacity=0.3] (i0) to (i2);

            \tikzset{xshift=1.3cm, yshift=-4cm}

                \node (a1) at (-1,-0.5) [vc, \colorg, fill=\colorg!10] {$\beta$};
                \node (a3) at (1,-0.5) [v2, \colorf] {$\zeta$};
                \node (a2) at (0.3,1.1) [v2, \colord, scale=1.15] {$\gamma$};
        
                \draw[line width=1.8, garrow2, opacity=1, \colorc] plot [smooth] coordinates {(-.4,-.1) (0.2, 0.6) (0.6, -.2) (-0.2,-.2)};
                \draw[very thick, garrow] (a1) to (a2);
                \draw[very thick, garrow] (a2) to (a3);
                \draw[very thick, garrow] (a3) to (a1);
                
                \draw[very thick, dashed] (-1.4,-0.5) to (-1.8,-0.5);
                \draw[very thick] (-1.4,-0.5) to (a1);

            \tikzset{xshift=-1.3cm, yshift=+4cm}

            \draw[line width=1.8, garrow3] (3.8, -2.5) to (5.7, -2.5);
            
            \tikzset{xshift=8cm}

            \node (a0) at (0,0) [v, \colore, label={160:$v$}, scale=1.15] {$\alpha$};
    
            \node (a4) at (-45:1.5cm) [v, \colord, scale=1.15] {$\gamma$};
            \node (a5) at (-90:1.5cm) [v, \colorg, label={180:$w$}] {$\beta$};
            \node (a3) at (0:1.5cm) [v, \colorf, label={20:$x$}] {$\zeta$};
            \node (a9) at (-135:1.5cm) [vertex, label={160:$u$}] {};

            \draw[very thick, dashed, opacity=.6] (a0) to (180:.6cm);
            \draw[very thick, \colorf] (a0) to (a4);
            \draw[very thick, \colord] (a0) to (a5);
            \draw[line width = 1.9, \colorg] (a0) to (a3);
            \draw[line width = 1.9, \colorg] (a0) to (a9);

            \node (h1) at (-60:.9cm) [labell, \colorf] {$\zeta$};
            \node (h2) at (-105:.9cm) [labell, \colord] {$\gamma$};
            \node (h3) at (-15:.9cm) [labell, \colorg] {$\beta$};
            \node (h9) at (-150:.9cm) [labell, \colorg] {$\beta$};

              \node (b1) at (-100:2.5cm) [vertex] {};
            \node (b2) at (-125:2.5cm) [vertex] {};
            
            \draw[line width = 1.9, \colore] (a5) to (b1);
                \draw[dashone, line width = 1.9, draw=\colorg, postaction={draw=\colore, dash phase= \dash}] (b1) to[bend left=20] (b2);
                
            \draw[line width = 1.9, \colore] (b2) to (a9);

            \node (h0) at (-1,-3.1) [labell, scale=1] {$\sigma'$};
            \node (h0) at (2.5,-2) [labell, scale=1] {$\cD_{\sigma'}(v)$};

            \node (i0) at (-1.5,-5) [labell] {};
            \node (i2) at (3,0) [labell] {};
            \draw[thick, opacity=0.3] (i0) to (i2);

            \tikzset{xshift=1.3cm, yshift=-4cm}

                \node (a1) at (-1,-0.5) [vc, \colorg, fill=\colorg!10] {$\beta$};
                \node (a3) at (1,-0.5) [v2, \colorf] {$\zeta$};
                \node (a2) at (0.3,1.1) [v2, \colord, scale=1.1] {$\gamma$};
        
                \draw[very thick, garrow] (a2) to (a1);
                \draw[very thick, garrow] (a3) to (a2);
                
                \draw[very thick, dashed] (-1.4,-0.5) to (-1.8,-0.5);
                \draw[very thick] (-1.4,-0.5) to (a1);
                \draw[very thick, garrow] (a1) to (a3);

            \tikzset{xshift=-1.3cm, yshift=+4cm}

            \end{tikzpicture}

        \caption{Illustration of \Cref{lem:circuitone}. The color $\beta$ is the only marked color in the directed cycle $\C = (\beta, \gamma, \zeta)$. Around $v$, we only represent the color in $\C$. We recolor the edges incident to $v$ from the rightmost edge colored $\beta$, then follow the cycle in the color shift digraph around $v$.  }
        \label{fig:movearound}
\end{figure}
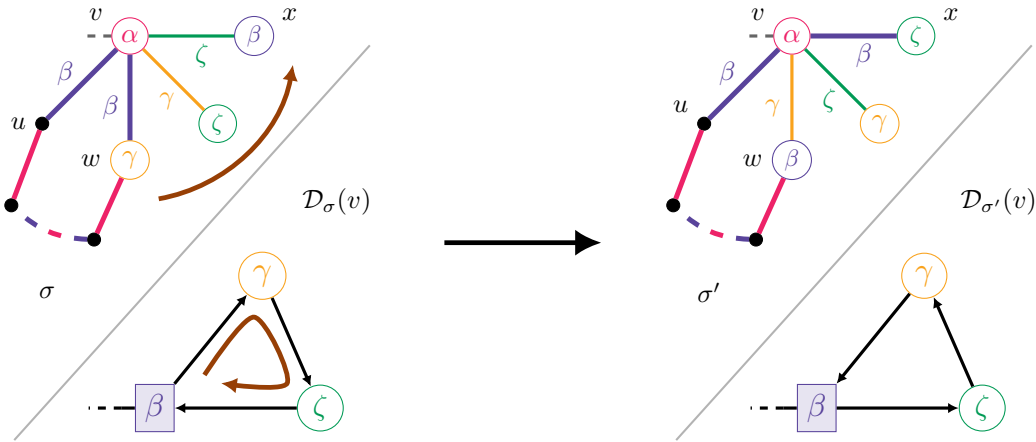

\section{Concluding remarks}\label{sec:concl}
In this paper, we take a first step toward efficient sampling of proper edge-colorings with an optimal number of colors by analyzing the randomized hill-climbing algorithm of Dotan, Linial, and Peled~\cite{Dotan20} on arbitrary graphs.
\smallskip

\noindent\textbf{Efficiency of the algorithm.} A straightforward analysis shows that the expected running time of the algorithm is $\exp(O(n))$. However, simulations suggest that the expected running time is polynomial, aligning with the $\tilde{O}(n^4)$ conjecture in~\cite{Dotan20} for even complete graphs. In fact, the algorithm only takes exponential time during a specific step; see \Cref{lem:bichromaticcycle}. 
\smallskip

\noindent\textbf{Uniformity of the target distribution.} For even complete graphs, it is conjectured that starting from a coloring obtained by assigning colors to edges independently and uniformly at random, the algorithm converges to the uniform distribution over all proper edge-colorings. Extending this result to general graphs is a much broader challenge, and we do not venture to propose such a conjecture at this stage.
\smallskip

\noindent\textbf{Class I graphs.} While the algorithm is analyzed for an optimal number of colors in Class II graphs, it is natural to ask about its behavior in Class I graphs when restricted to only $\Delta$ colors. In this setting, efficient sampling in general is impossible, but the question of whether the algorithm terminates almost surely is still valid. The case of even complete graphs is  a specific instance of this question. 
\smallskip

\noindent\textbf{Design theory.} We also highlight the related problem of efficiently sampling combinatorial designs, such as Latin squares, 1-factorizations of complete graphs, or Steiner triple systems. Despite the growing interest in random combinatorial designs in recent years, efficient algorithms to generate them remain elusive.
\smallskip

\noindent\textbf{Beyond Glauber dynamics.} Glauber dynamics may fail to be ergodic (i.e., the recoloring graph is not connected) when fewer than $2\Delta$ colors are used. The challenge of designing efficient approximate sampling algorithms below this threshold remains a wide-open problem.

\bigskip

\textbf{Acknowledgement.} The authors would like to thank the organisers of the Villeurbanne Reconfiguration Workshop at the Université Claude Bernard Lyon 1 for their warm hospitality, where this project got started.

\bibliographystyle{plain}
\bibliography{biblio}

\end{document}